\documentclass[a4paper,conference]{IEEEtran}
\usepackage{ifthen}
\usepackage{graphicx}
\usepackage{epsfig}
\usepackage{amsfonts,dsfont,amssymb,amsthm,stmaryrd}
\usepackage[cmex10]{amsmath} 

\newcommand\independent{\protect\mathpalette{\protect\independent}{\perp}} 
\def\independent#1#2{\mathrel{\rlap{$#1#2$}\mkern2mu{#1#2}}}

\newcommand{\supp}{\mathrm{Supp}}

\newcommand{\F}{\mathbb{F}}  
\newcommand{\mR}{\mathbb{R}} 

\newcommand{\mZ}{\mathbb{Z}}

\newcommand{\pp}{\mathbb{P}}
\newcommand{\E}{\mathbb{E}}

\newcommand{\e}{\varepsilon}

\newcommand{\rank}{\mathrm{rank}}

\newcommand{\A}{\mathcal{A}}

\newcommand{\Y}{\mathcal{Y}}

\newcommand{\jj}{[j]}
\newcommand{\ii}{[i]}

\newcommand{\J}{[J]}
\newcommand{\C}{[J^c]}
\newcommand{\K}{[J\setminus i]}
\newcommand{\m}{[m]}
\newcommand{\1}{[1]}
\newcommand{\M}{E_m}
\newcommand{\bX}{\bar{X}}

\newcommand{\bU}{\bar{U}}
\newcommand{\mat}{\mathrm{MAT}}

\theoremstyle{definition}
\newtheorem{definition}{Definition}
\theoremstyle{plain}
\newtheorem{thm}{Theorem}
\theoremstyle{plain}
\theoremstyle{plain}
\newtheorem{lemma}{Lemma}
\theoremstyle{plain}
\newtheorem{corol}{Corollary}
\theoremstyle{plain}
\theoremstyle{remark}
\theoremstyle{plain}

\begin{document}
\title{Polar Codes for the $m$-User MAC}
\author{%
\IEEEauthorblockN{Emmanuel Abbe, Emre Telatar}
\IEEEauthorblockA{Information Processing Group, EPFL\\
Lausanne 1015, Switzerland \\
Email: \{emmanuel.abbe,emre.telatar\}@epfl.ch\\
              }
}

\maketitle


\begin{abstract}
In this paper, polar codes for the $m$-user multiple access channel (MAC) with binary inputs are constructed.  
It is shown that Ar{\i}kan's polarization technique applied individually to each user transforms independent uses of a $m$-user binary input MAC into successive uses of extremal MACs. This transformation has a number of desirable properties: (i) the `uniform sum rate'\footnote{In this paper all mutual informations are computed when the inputs of a MAC are distributed independently and uniformly.   The resulting rate regions, sum rates, etc., are prefixed by `uniform' to distinguish them from the capacity region, sum capacity, etc.} of the original MAC is preserved, (ii) the extremal MACs have uniform rate regions that are not only polymatroids but matroids and thus (iii) their uniform sum rate can be reached by each user transmitting either uncoded or fixed bits; in this sense they are easy to communicate over.  
A polar code can then be constructed with an encoding and decoding complexity of $O(n \log n)$ (where $n$ is the block length), a block error probability of $o(\exp(- n^{1/2 - \e}))$, and capable of achieving the uniform sum rate of any binary input MAC with arbitrary many users. An application of this polar code construction to communicating on the AWGN channel is also discussed. 
\end{abstract}

\section{Introduction}
The polarization technique, introduced by Ar{\i}kan in \cite{ari}, transforms $n$ independent uses of a noisy binary input channel into single-uses of $n$ synthetic binary input channels. The key property of this transformation is that almost all of these synthetic channels are polarized, in the sense that they are either very noisy or almost noiseless (i.e., having a mutual information either close to 0 or to 1). 
Moreover, this technique preserves the `uniform mutual information' --- the mutual information of the channel with the uniform input distribution --- that is, the proportion of synthesized channels that are almost noiseless tends to the uniform mutual information. 
As the very noisy or almost noiseless channels are channels for which it is easy to code, 
this transformation leads to the following coding scheme:
uncoded information bits are sent on the polarized channels that have uniform mutual informations close to 1, 
and on the remaining channels, bits frozen to pre-determined values are transmitted.


In addition to bringing a new perspective on coding, polar codes can be implemented with low computational effort. More precisely, the encoding and decoding complexity of a polar code is $O(n \log n)$. 
By definition of the uniform mutual information, these codes achieve the capacity of any channel whose capacity achieving input distribution is uniform. 
The original polar code construction was generalized in \cite{qpol} for channels with binary input alphabets to channels with alphabets of arbitrary prime cardinality, allowing polar codes to approach the capacity of any discrete memoryless channel.

In this paper, we show how the polarization technique can be extended to a multi-user problem, namely, the multiple access channel (MAC). One of the interesting aspect of this extension is that, as opposed to the single-user setting where a single mutual information characterizes an achievable rate, there is in a MAC setting a collection of mutual informations that characterize an achievable rate {\it region}. Hence, the terminology ``polarized'' may need to be revised in a MAC setting, as there may be more than two ``polarized MACs". 
Indeed, for a 2-user binary input MAC, by applying Arikan's construction to each user's input separately, \cite{2mac} shows that $n$ independent uses of a 2-user MAC are converted into $n$ successive uses of five possible ``extremal 2-user MACs". These 2-user extremal MACs are: (i) each user sees a pure noise channel, (ii) one of the user sees a pure noisy channel and the other sees a noiseless channel, (iii) both users see a noiseless channel, (iv) a pure contention channel: a channel whose uniform rate region is the triangle with vertices (0,0), (0,1), (1,0). Note that for this channel, if any of the two users communicates at zero rate, the other user sees a noiseless channel. Moreover \cite{2mac} shows that the uniform sum rate of the original MAC is preserved during the polarization process, and that the polarization to the extremal MACs occurs fast enough, so as to ensure the construction of a polar code with vanishing block error probability, achieving uniform sum rate on binary inputs 2-user MACs.  


In contrast to \cite{2mac}, here we investigate the polarization of the MAC for an arbitrary number of users. 
In the two user case, the extremal MACs are not just MACs for which each users sees either a noiseless or pure noise channel, as there is also the pure contention channel. However, the uniform rate region of the 2-user extremal MACs are all polyhedrons with integer valued constraints. So, the first interesting aspect of the polarization of the MAC with arbitrary many users is to understand what pattern do extremal MACs follow. We will see that the 2-user and 3-user cases can be handled in a similar manner, whereas a new phenomenon appears when the number of users reaches 4, and the extremal MACs are no longer in a one to one correspondence with the polyhedrons having integer valued constraints.
To characterize the extremal MACs, we first show how a relationship between random variables defined in terms of mutual information falls precisely within the independence notion of the matroid theory. This connection is used to show that the extremal MACs are in a one to one correspondence with binary matroids, and are ``equivalent'' (in a sense which will be defined later) to linear deterministic MACs. This is then used to conclude the construction of a polar code ensuring reliable communication on binary input MACs for arbitrary values of $m$.


Finally, we discuss two applications resulting from the MAC polar code construction with arbitrary many users described in this paper. The first one is motivated by the idea of proposing a new coding scheme for the additive white Gaussian noise channel. By transmitting the standardized average of $m$ binary inputs which are uniformly distributed (taking into account the power constraint), we transmit a random input which is approximately Gaussian distributed when $m$ is large (using the central limit theorem). This is important to achieve the highest rate on the AWGN channel, since the Gaussian input distribution maximizes the mutual information for this channel. We can then use the polar code construction for a MAC developed in this paper to propose a new coding scheme for the AWGN channel.
In the second application, we construct polar codes achieving the uniform sum-rate of MACs with $q$-ary inputs, where $q$ is a power of 2, using the polar code construction for MACs with binary inputs and a large enough number of users. We also show how, with this extension, the sum-capacity of any $m$-user MAC can be achieved. 
%

\section{Polarization Process for MACs}\label{cons}

We consider a $m$-user multiple access channel with binary input alphabets (BMAC) and arbitrary output alphabet $\Y$. The channel is specified by the conditional probabilities
$$
P(y|\bar{x}),\quad\text{for all $y \in \Y$ and $\bar{x}=(x\1,\ldots,x\m) \in \F_2^m$}.
$$

Let $\M := \{1,\ldots,m\}$ and let $X\1,\ldots,X\m$ be mutually independent and uniformly distributed binary random variables. Let $\bX:=(X\1,\ldots,X\m)$. We denote by $Y$ the output of the MAC $P$ when the input is $\bX$. 
For $J \subseteq \M$, we define
\begin{align*}
 X\J & := \{ X\ii  : i \in J \}, \\
 I\J(P) & := I(X\J;Y X\C),
\end{align*}
where $J^c$ denotes the complement set of $J$ in $\M$, and
\begin{align}
 I(P)  :\, 2^{E_m} & \rightarrow \mR \notag \\
 J \,\,& \mapsto I\J(P) \label{Ifct}
\end{align}
where $2^{E_m}$ denotes the power set of $\M$ and where $I[\emptyset](P) = 0$.
Note that 
\begin{align*}
& \mathcal{I}(P) := 
 \{(R_1,\ldots,R_m): \, 0 \leq \sum_{i \in J} R_i \leq I\J(P), \, \forall J \subseteq \M \}
\end{align*}
is included in the capacity region of the MAC $P$. We refer to $\mathcal{I}(P)$ as the uniform rate region and to $I[\M](P)$ as the uniform sum rate. 
We now consider two independent uses of such a MAC. We define
\begin{align*}
& \bX_1:= (X_1\1,\ldots, X_1\m), \quad  \bX_2 := (X_2\1,\ldots, X_2\m),  
\end{align*}
where $X_1\ii,X_2\ii$, with $i \in \M$, are mutually independent and uniformly distributed binary random variables. We denote by $Y_1$ and $Y_2$ the respective outputs of independent uses of the MAC $P$ when the inputs are $\bX_1$ and $\bX_2$:
\begin{align}
& \bX_1 \stackrel{P}{\rightarrow} Y_1,  \quad
\bX_2 \stackrel{P}{\rightarrow} Y_2. \label{p1}
\end{align} 
We define two additional binary random vectors 
\begin{align*}
& \bU_1 :=(U_1\1, \ldots ,U_1\m) , \quad \bU_2:=(U_2\1,\ldots,U_2\m)
\end{align*} 
with mutually independent and uniformly distributed components, and we put $\bX_1$ and $\bX_2$ in the following one to one correspondence with $\bU_1$ and $\bU_2$:
\begin{align*}
\bX_1 = \bU_1 + \bU_2,\qquad
\bX_2 =  \bU_2,  
\end{align*}
where the addition in the above is the modulo 2 component wise addition. 

\begin{definition}
Let $P: \F_2^m \rightarrow \Y$ be a $m$-user BMAC.
We define two new $m$-user BMACs, $P^-: \F_2^m \rightarrow \Y^2$ and $P^+: \F_2^m \rightarrow  \Y^2 \times \F_2^m$, by
\begin{align*}
& P^- (y_1,y_2|\bar{u}_1) := \sum_{\bar{u}_2 \in \F_2^m} \frac{1}{2^m} P(y_1|\bar{u}_1+\bar{u}_2) P(y_2|\bar{u}_2) ,\\
& P^+ (y_1,y_2, \bar{u}_1 | \bar{u}_2) :=  \frac{1}{2^m} P(y_1|\bar{u}_1+\bar{u}_2) P(y_2|\bar{u}_2),
\end{align*}
for all $\bar{u}_i \in \F_2^m$, $y_i \in \Y$, $i=1,2$.
\end{definition}
That is, 
we have now two new $m$-user BMACs with extended output alphabets: 
\begin{align}
& \bU_1  \stackrel{P^-}{\rightarrow} (Y_1, Y_2), \quad 
\bU_2  \stackrel{P^+}{\rightarrow} (Y_1, Y_2, \bU_1) \label{p+}
\end{align} 
which also defines 
$I\J(P^-)$ and $I\J (P^+)$, $\forall J \subseteq \M $. 

This construction is the natural extension of the construction for $m=1,2$ in \cite{ari,2mac}. Here again,  we are comparing two independent uses of the same channel $P$ (cf. \eqref{p1}) with two successive uses of the channels $P^-$ and $P^+$ (cf. \eqref{p+}). Note that $$I\J(P^-) \leq I\J (P) \leq I\J (P^+),\quad \forall J \subseteq \M .$$ 

\begin{definition}\label{rp}
Let $\{B_n\}_{n \geq 1}$ be i.i.d. uniform random variables valued in $\{-,+\}$. Let the BMAC valued random process $\{P_n, \ n \geq 0\}$ be defined by 
\begin{align}
& P_0 := P, \notag \\
& P_n := P_{n-1}^{B_n}, \quad \forall n \geq 1. \label{ind} 
\end{align}
\end{definition}

\subsection{Discussion}
When $m=1$, we have $2 I(P) = I(P^-) + I(P^+)$, which implies that $I(P_n)$ (which in this case denotes a sequence of scalar random variables and not of functions) is a martingale. 
This allows to show that $I(P_n)$ tends to either 0 or 1, and the extremal channels of the single-user polarization scheme are either pure noise or noiseless channels.
Moreover, in the polarization of the single-user channel, the extremal channels are synthesized by using a genie aided decoder.  The genie helps the decoder in providing the correct values of the previous decisions when decoding the current channel's input.  In the polar code construction the genie is simulated by a decoder which decodes the bits successively on the synthetic channels, and uses its previous decisions assuming they are correct.  As the block error probability of the genie aided and the standalone decoder are exactly the same, it is sufficient to study the block error probability of the genie aided decoder. 
These facts then facilitates the design of a code: bits are frozen on the very noisy channels and uncoded information bits are sent on the almost noiseless channels, recovered then by using a successive decision decoder at the receiver. To show that the block error probability of this coding scheme is small, i.e., that the error caused by the successive decision decoder does not propagate, it is shown that the convergence to the ``good'' extremal channels is fast enough.  
When $m \geq 2$, several new points need to be investigated. In particular, one needs to check whether $I\J(P_n)$ still has a martingale property for different $J$'s. 
Then, if the convergence of each $I\J(P_n)$ can be proved, one has to examine whether the obtained limiting MACs are also extremal MACs, along the spirit of creating trivial channels to communicate over, as in the single-user polarization. 
Finally, one needs to ensure that the convergence of these mutual informations is taking place fast enough, 
so as to ensure a block error probability that tends to zero when successive decision decoding is used.

\section{Preliminary Results}
Summary: In Section \ref{lim}, we show that $I(P_n)$ tends a.s.\ to a matroid rank function (cf.~Definition \ref{matdef}). 
We then see in Section \ref{comm} that the extreme points of a uniform rate region with matroidal constraints can be achieved by each user sending uncoded or frozen bits; in particular the uniform sum rate can be achieved by such strategies.
We then show in Section \ref{msol}, that for arbitrary $m$, $I(P_n)$ tends not to an arbitrary matroid rank function but to the rank function of a binary matroid (cf.~Definition \ref{bin}). This is used to show that the convergence to the extremal MACs happens fast enough, which then leads to the main result of this paper, Theorem \ref{main} in Section \ref{msol}.  This theorem states that applying Ar{\i}kan's polar transform separately to each user, and using a successive decision decoder can achieve sum rates arbitrarily close to the uniform sum rate of a MAC, ensure block error probability that decays roughly like $2^{-\sqrt{n}}$ with block length, and operate with computational complexity $O(n\log n)$.

 

\subsection{The extremal MACs}\label{lim}

\begin{lemma}\label{mart}
$\{I\J(P_n) ,  n \geq 0\}$ is a bounded super-martingale when $J \varsubsetneq \M $ and a bounded martingale when $J = \M$. 
\end{lemma}
\begin{proof}
For any $J\subseteq \M$, $I\J(P_n) \leq m$ and
\begin{align}
 2 I\J (P) &= I( X_1\J X_2\J ; Y_1 Y_2 X_1\C X_2\C)  \notag\\
&=I(U_1\J U_2\J ; Y_1 Y_2  U_1\C U_2\C) \notag \\
&= I(U_1\J ;Y_1 Y_2 U_1\C U_2\C)\notag \\&\quad + I(U_2\J ; Y_1 Y_2 U_1\C U_2\C U_1\J)\notag \\
&\geq I(U_1\J  ;Y_1 Y_2 U_1\C )\notag \\&\quad + I(U_2\J ; Y_1 Y_2 \bU_1 U_2\C) \notag \\
&= I\J(P^-) + I\J (P^+) \label{sup}.
\end{align}
If $J=\M$, the inequality above is an equality.
\end{proof}
Note that the inequality in the above are only due to the bounds on the mutual informations of the $P^-$ channel.  
Because of the equality when $J=\M$, our construction preserves the uniform sum rate.
As a corollary of previous Lemma, we have the following result. 
\begin{lemma}\label{conv}
The process $\{ I\J(P_n), J \subseteq \M\}$ converges a.s., i.e., for each $J \subseteq \M$, $\lim_{n \rightarrow \infty} I\J (P_n)$ exists a.s..
\end{lemma}
Note that for a fixed $n$, $\{ I\J(P_n), J \subseteq \M\}$ denotes the collection of the $2^m$ random variables $I\J(P_n)$, for $J \subseteq \M$.
When the convergence takes place (this is an a.s. event), let us define 
$$I_\infty \J := \lim_{n \rightarrow \infty} I\J(P_n)$$ 
and $I_\infty$ to be the function $J \mapsto I_\infty \J$.

From previous theorem, $I_\infty \J$ is a random variable valued in $[0, |J|]$.  
We will now further characterize these random variables. 

The following Lemma is proved in \cite{ari}.
\begin{lemma}\label{bjg}
For any $\e >0$, there exists $\delta>0$ such that 
$I(A_2;B_1 B_2 A_1)- I(A_2;B_2) < \delta$ implies  $$I(A_2;B_2) \in [0,\e) \cup (1-\e,1],$$
whenever $(A_1,A_2,B_1,B_2)$ are random variables valued in $\F_2 \times \F_2 \times \mathcal{B} \times \mathcal{B}$, with $\mathcal{B}$ any set, and
$$\pp_{A_1 A_2 B_1 B_2} (a_1,a_2,b_1,b_2) = \frac{1}{4} Q(b_1|a_1+ a_2) Q(b_2|a_2), $$
for any $a_i \in \F_2$, $b_i \in \mathcal{B}$, $i=1,2$, where $Q$ is a binary input $\mathcal{B}$-output channel.
\end{lemma}
This Lemma is used to prove the following. 
\begin{lemma}\label{mjg}
For any $\e>0$ and any $m$-user BMAC $P$, there exists $\delta>0$, such that for any $J\subseteq \M$, if $I\J(P^+) - I\J (P) < \delta$, we have $$ I\J (P) - I\K (P) \in [0,\e) \cup (1-\e,1],\quad \forall i \in J,$$ where $ I [\emptyset] (P)= 0$.
\end{lemma}
\begin{proof} 
Let $i \in J$. Note that
\begin{align*}
& I\J(P^+) - I\J (P) \\
&= I(U_2\J; Y_1 Y_2 \bU_1  U_2\C) - I(U_2\J;Y_2  U_2\C)\\
&= I(U_2\J; Y_1 \bU_1 |Y_2  U_2\C)\\
&\geq I(U_2\ii; Y_1 U_1\ii U_1\C |Y_2  U_2\C)\\
&= I(U_2\ii; Y_1 U_1\C Y_2 U_2\C  U_1\ii ) - I(U_2\ii;Y_2 U_2\C) \\
&= I(U_2\ii; Y_1 X_1\C Y_2 X_2\C  U_1\ii ) - I(U_2\ii;Y_2 X_2\C) .
\end{align*}
Using Lemma \ref{bjg} with $A_k=U_k\ii$, $B_k=Y_k  X_k\C$, $k=1,2$, and 
$$Q(y , x\C | x\ii)= \frac{1}{2^{m-1}} \sum_{x\jj \in \F_2, j\notin J^c \cup \{i\} } P(y|\bar{x}), $$ we conclude that we can take $\delta$ small enough, so that $I\J(P^+) - I\J (P) < \delta$ implies $ I(U_2\ii;Y_2 X_2\C) \in [0,\e) \cup (1-\e,1] $.
Moreover, we have
$$ I\J(P) = I\K(P) + I(U_2\ii;Y_2X_2\C). \qedhere$$
\end{proof}

\begin{lemma}\label{01}
With probability one, 
$I_\infty \J - I_\infty \K \in \{0,1\}$, $\forall J \subseteq \M, i \in J $, where $I_\infty [\emptyset]:=0$.
\end{lemma}
\begin{proof}
From Lemma \ref{conv}, we have that $I[J] (P_n)$ converges a.s., hence $\lim_{n \rightarrow \infty} |I[J] (P_{n+1}) - I[J] (P_n)|=0$ a.s. \, Moreover, by definition of $P_n$, $|I[J] (P_{n+1}) - I[J] (P_n)|$ is equal to $I[J](P_n^+)-I[J](P_n)$ w.p. half and $I[J](P_n)-I[J](P_n^-)$ w.p. half. Hence, from \eqref{sup}, $\E |I[J] (P_{n+1}) - I[J](P_n)| \geq \frac{1}{2}( I[J](P_n^+)-I[J](P_n))$. But $|I[J] (P_{n+1}) - I[J] (P_n)|$ is bounded by $m$, hence $\lim_{n \rightarrow \infty} \E |I[J] (P_{n+1}) - I[J] (P_n)| =0$ and $\lim_{n \rightarrow \infty}I[J](P_n^+)-I[J](P_n)=0$. Finally, we conclude using Lemma \ref{mjg}.
\end{proof}

Note that Lemma \ref{01} implies in particular that 
$\{ I_\infty \J, J \subseteq \M\}$ is a.s. a discrete random vector.
\begin{definition}
We denote by $\mathcal{A}_m$ the support of $\{ I_\infty \J, J \subseteq \M\}$ (when the convergence takes place, i.e., a.s.). This is a subset of $\{0,\ldots,m\}^{2^m}$.
\end{definition}
We have already seen that not every element in $\{0,\ldots,m\}^{2^m}$ can belong to $\A_m$. We will now further characterize the set $\A_m$.
\begin{definition}
A polymatroid is a set $\M$, called a ground set, equipped with a function $f: 2^m \rightarrow \mR$ (where $2^m$ denotes the power set of $\M$), called a rank function, which satisfies 
\begin{align*}
& f(\emptyset)=0,\\
&f[J] \leq f[K], \quad  \forall J \subseteq K \subseteq \M, \\
&f[J \cup K]  + f[J \cap K] \leq f[J] + f[K] , \quad \forall J,K \subseteq \M.
\end{align*}
\end{definition}

A proof of the following result can be found in \cite{hanly}. 
\begin{thm}
For any MAC and any distribution of the inputs $X[E]$, we have that 
$\rho(S)=I(X[S];YX[S^c])$ is a rank function on $E$,
where we denote by $Y$ the output of the MAC when the input is $X[E]$. Hence, $(E,\rho)$ is a polymatroid.
\end{thm}

Therefore, any realization of $I(P_n) $ is a rank function and the elements of $\mathcal{A}_m$ are the image of a polymatroid rank function. 
\begin{definition}\label{matdef}
A matroid is a polymatroid whose rank function is integer valued and satisfies $f(J) \leq |J|$, $\forall J \subseteq \M$. We denote by $\mathrm{MAT}_m$ the set of all matroids with ground state $\M$. We use the notation $r_\mathbb{B}$ to refer to the rank function of a matroid $\mathbb{B}$. We will sometimes identify a matroid with its rank function image, in which case, we consider an element of $\mathrm{MAT}_m$ as a $2^m$ dimensional integer valued vector. 
We also define a basis of a matroid by the collection of maximal subsets of $\M$ for which $f(J) = |J|$. 
\end{definition}
Using Lemma \ref{01} and the definition of a matroid, we have the following result.

\begin{thm}\label{mat}
For every $m \geq 1$, $\mathcal{A}_m \subseteq \mathrm{MAT}_m,$ i.e., $I_\infty$ is a matroid rank function.
\end{thm}
We will see that the inclusion is strict for $m \geq 4$.

\subsection{Communicating on MACs with matroidal regions}\label{comm}

We have shown that, when $n$ tends to infinity, the MACs that we create with the polarization construction of Section \ref{cons} are particular MACs: the mutual informations $I_\infty \J$ are integer valued (and satisfy the other matroid properties). A well-known result in matroid theory (cf. Theorem 22 of \cite{edm}) says that the vertices of a polymatroid given by a rank function $f$ are the vectors of the following form:
\begin{align*}
& x_{j(1)} = f(A_1),\\
& x_{j(i)} = f(A_i)-f(A_{i-1}), \quad \forall 2 \leq i \leq k \\
& x_{j(i)}=0, \quad \forall k < i \leq m,
\end{align*}
for some $k \leq m$, $j(1),j(2),\ldots,j(m)$ distinct in $\M$ and $A_i = \{j(1),j(2),\ldots,j(i)\}$, where the vertices strictly in the positive orthant are given for $k=m$.

Therefore, we have the following corollary.
\begin{corol}
\label{cor:01}
The uniform rate region defined by an element of $\mathcal{A}_m$ has vertices on the hypercube $\{0,1\}^m$.  In particular, to communicate at a rate $m$-tuple which is a vertex of such a MAC uniform rate region, each user communicates on either a noiseless or pure noise channel.
\end{corol}



\subsection{Convergence Speed and Representation of Matroids}\label{speed}
{\it Convention:} for a given $m$, we write the collection
 $\{ I_\infty \J , J \subseteq \M\}$ by skipping the empty set (since $I_\infty [\emptyset] =0$) as follows: when $m=2$, we order the sequence as $(I_\infty [1], I_\infty [2], I_\infty [1,2])$, and when $m=3$, as $(I_\infty [1], I_\infty [2],$ $I_\infty [3],$ $I_\infty [1,2], I_\infty [1,3], I_\infty [2,3], I_\infty [1,2,3])$, etc.

In this section, we show that there is a correspondence between the extremal MACs and the linear deterministic MACs, i.e., MACs whose outputs are linear forms of the inputs. 
This correspondence has been used in \cite{2mac} to establish that the convergence to the extremal MACs for the 2-user case is fast, namely $o(2^{-n^{\beta}})$ for any $\beta< 1/2$, which allows to conclude that the block error probability of the code described in \cite{2mac} is small. We hence follow the same approach as in \cite{2mac} to treat the case where the number of users is arbitrary, and proceed here to establish this correspondence. We will see that while the case $m=3$ is similar to the case $m=2$, a new difficulty is encountered for $m \geq 4$. How to use this correspondence in order to show that the the convergence to the extremal MACs for the $m$-user case is fast is done in Section \ref{msol}.

Note that a property of the matroids $\{(0,0,0),(0,1,1)$, $(1,0,1)$,$(1,1,1),(1,1,2) \}$ is that we can express any of them as the uniform rate region of a linear deterministic MAC: $(1,0,1)$ is in particular the uniform rate region of the MAC whose output is $Y=X[1]$, $(0,1,1)$ corresponds to $Y=X[2]$, $(1,1,1)$ to $Y=X[1]+X[2]$ and $(1,1,2)$ to $Y=(X[1],X[2])$. Indeed, this is related to the fact that any matroid with a two element ground state can be represented in the binary field. Let us introduce the definition of binary matroids. 
\begin{definition}\label{bin}
{\it Linear matroids:} let $A$ be a $k \times m$ matrix over a field. Let $\M$ be the index set of the columns in $A$. The rank of $J \subseteq \M$ is defined by the rank of the sub-matrix with columns indexed by $J$. \\
{\it Binary matroids:} A matroid is binary if it is a linear matroid over the binary field. We denote by $\mathrm{BMAT}_m$ the set of binary matroids with $m$ elements.
\end{definition}

\subsubsection{The Case $m=3$}
$\mat_3$ is given by
8 unlabeled matroids (16 labeled ones). Moreover, they are all binary representable (there are 16 labeled binary matroids). 
For example, 
it is clear that the deterministic MAC whose output is $X[1]+X[2]+X[3]$ has a uniform rate region given by $(1,1,1,1,1,1,1)$. Similarly, all matroids for $m=3$ correspond to the rate region of a linear deterministic MAC. 
However, one can also show that any 3-user binary MAC with uniform rate region given by a matroid is equivalent to a linear deterministic MAC in the following sense. A MAC with output $Y$ and uniform rate region given by $(1,1,1,1,1,1,1)$ must satisfy $I(X[1]+X[2]+X[3]; Y)=1$, and similarly for other matroids (with $m=3$), where the linear forms of inputs which can be recovered from the output are dictated by the binary representation of the matroid.
However, the above claim is not quite sufficient to show that, if $\{I\J(P_n) , J \subset \M\}$ tends to $(1,1,1,1,1,1,1)$, we have along this path that $I((P^{[1,2,3]})_n)$ tends to 1, where $P^{[1,2,3]}$ is the channel with input $X[1]+X[2]+X[3]$ and output $Y$. For this, one can show a stronger version of the claim which says that if a MAC has a uniform rate region ``close to'' $(1,1,1,1,1,1,1)$, it must be that $I(X[1]+X[2]+X[3]; Y)$ is close to 1.
In any case, a similar technique as for the $m=2$ case lets one show that the convergence to the matroids in  $\A_3$ must take place fast enough. 

\subsubsection{The Case $m=4$}
We have that $\mat_4$ contains 17 unlabeled matroids (68 labeled ones). However, there are only 16 unlabeled binary matroids with ground state 4. Hence, there must be a matroid which does not have a binary representation. This matroid is given by $(1,1,1,1,2,2,2,2,2,2,2,2,2,2)$ (one can easily check that this is not a binary matroid). It is denoted $U_{2,4}$ and is called the uniform matroid of rank 2 with 4 elements (for which any 2 elements set is a basis). Of course, that this matroid is not binary, does not imply that an hypothetic convergence to it must be slow. It means that we will not be able to use the technique employed for the case $m=2,3$.

Luckily, one can show that there is no MAC leading to $U_{2,4}$ and the following holds. 
\begin{lemma}\label{spock}
$\A_4 \subset \mathrm{BMAT}_4 \subsetneq \mat_4.$
\end{lemma}
Hence, the $m=4$ case can be treated in a similar manner as the previous cases. 
We conclude this section by proving the following result, which implies Lemma \ref{spock}.
\begin{lemma}\label{u24}
$U_{2,4}$ cannot be the uniform rate region of any MAC with four users and binary inputs.
\end{lemma}
\begin{proof}
Assume that $U_{2,4}$ is the uniform rate region of a MAC. We then have
\begin{align}
& I(X[i,j]; Y) = 0 , \label{c1} \\ 
& I(X[i,j]; Y X[k,l]) = 2,  \label{c2}
\end{align}
for all $i,j,k,l$ distinct in $\{1,2,3,4\}$.


Let $y_0$ be in the support of $Y$. For $x\in \F_2^4$, define $\pp(x|y_0) = W(y_0|x)  / \sum_{z \in \F_2^4}W(y_0|z)$.
Assume w.l.o.g. that $p_0:=\pp(0,0,0,0|y_0)>0$. 
Then from \eqref{c2}, $\pp(0,0,*,*|y_0)=0$ for any choice of $*,*$ which is not $0,0$ and $\pp(0,1,*,*|y_0)=0$ for any choice of $*,*$ which is not $1,1$. On the other hand, from \eqref{c1},  $\pp(0,1,1,1|y_0)$ must be equal to $p_0$. However, we have form \eqref{c2} that $\pp(1,0,*,*|y_0)=0$ for any choice of $*,*$ (even for $1,1$ since we now have $\pp(0,1,1,1|y_0)>0$). At the same time, this implies that the average of $\pp(1,0,*,*|y_0)$ over $*,*$ is zero. This brings a contradiction, since from \eqref{c1}, this average must equal to $p_0$.
\end{proof}
Moreover, a similar argument can be used to prove a stronger version of Lemma \ref{u24} to show that no sequence of MACs can have a uniform rate region that converges to $U_{2,4}$.

\subsubsection{Arbitrary values of $m$}
We have seen in the previous section that for $m=2,3,4$, the extremal MACs have uniform rate region that are not any matroids but binary matroids. This fact can be used to show that for $m=2,3,4$, $\{ I\J(P_n) , J \subseteq \M\}$ must tend fast enough to $\{ I_\infty \J , J \subseteq \M\}$. The details of this proof are provided in Section \ref{msol}; in words, by working with the linear deterministic representation of the MACs, the problem of showing that the convergence speed is fast in the MAC setting becomes a consequence of a result shown in \cite{ari2} for the single-user setting.  We now show that the correspondence between extremal MACs and linear deterministic MACs holds for any value of $m$.
\begin{definition}
A matroid is BUMAC if its rank function can be expressed as $r (J) = I(X\J; Y X\C)$, $J \subseteq \M,$
where $X[E]$ has independent and binary uniformly distributed components, and $Y$ is the output of a binary input MAC with input $x[E]$. Note that the letters BU in BUMAC refer to the binary uniform (BU) inputs. 
\end{definition}

\begin{thm}\label{bumac}
A matroid is BUMAC if and only if it is binary. 
\end{thm}
The converse of this theorem is easily proved and the direct part, which can be found in \cite{ab}, uses the following theorem.
\begin{thm}[Tutte]
A matroid is binary if and only if it has no minor that is $U_{2,4}$.
\end{thm}

In the following theorem, we connect extremal MACs to linear deterministic MACs. 
\begin{thm}\label{with}
Let $X[E]$ have independent and binary uniformly distributed components.
Let $Y$ be the output of a MAC with input $X[E]$ and for which $f(J)=I(X\J; Y X\C)$ is integer valued, for any $J \subseteq \M$. Then, there exists a binary matrix $A$ such that
$$I(A X[E]; Y) = \mathrm{rank} A = f(\M).$$
\end{thm}

This theorem was originally proved using matroid theory notations and we refer to \cite{ab} for this proof and other investigations regarding the connection between matroid theory and extremal MACs. We provide an alternate proof of this theorem in the Appendix. One can also show a stronger version of this theorem for MACs having a uniform rate region which is close to a matroid, this is provided in the Theorem \ref{with2} below, whose proof is also given in the Appendix. Note that Theorem \ref{bumac} follows from Theorem \ref{with2}. 
\begin{thm}\label{with2}
Let $X[E]$ have independent and binary uniformly distributed components.
For any $\e>0$, there exists $\gamma(\e)$ with the following properties:
\begin{itemize}
\item[(i)] $\gamma(\e)\to0$ as $\e\to0$,
\item[(ii)]
Whenever $Y$ is the output of a MAC with input $X[E]$ and for which $f: 2^m \ni J \mapsto I(X\J; Y X\C)$ satisfies $\max_{J \in 2^m} d(f(J), \mZ) < \e$, there exists a binary matrix $A$ such that 
$$ | I(A X[E]; Y) - f(\M) | < \gamma (\e). $$
\end{itemize}
\end{thm}

Theorem \ref{bumac} says that an extremal MAC must have (with probability one) the same uniform rate region as the one of a linear deterministic MAC, i.e., a MAC whose output is a collection of linear forms of the inputs. However, Theorem \ref{with2}, says something stronger, namely, that from the output of an extremal MAC, one can recover a collection of linear forms of the inputs and essentially nothing else. In that sense, extremal MACs are equivalent to linear deterministic MACs. This also suggests that we could have started from the beginning by working with the quantities $I(P^{[J]}):=I(\sum_{i \in J} X_i;Y)$ instead of $I[J](P)=I(X\J; Y X\C)$ to analyze the polarization of a MAC. The second measure is the natural one to study a MAC, since it characterizes the rate region. However, we have just shown that it is sufficient to work with the first measure to characterize the uniform rate regions of the polarized MACs. Indeed, one can show that $I((P^{[J]})_n)$ tends either to 0 or 1 and Eren \c{S}a\c{s}o\u{g}lu \cite{eren} has provided a direct argument showing that these measures fully characterize the uniform rate region of the extremal MACs. We use a similar argument for the proof of Theorem \ref{with} given in the Appendix. 

\subsection{Comment: Relationship between information and matroid theories}
The process of identifying which matroids can have a rank function derived from an information theoretic measure, such as the entropy, has been investigated in different works, cf. \cite{zh} and references therein. 
In particular, the problem of characterizing the entropic matroids has consequent applications in network information theory and network coding problems as described in \cite{has}.

Entropic matroids are defined as follows. Let $E$ be a finite set and $X[E]=\{X_i\}_{i \in E}$ be a random vector with each component valued in a finite alphabet. Let $h(I) := h(X[I])$. 
\begin{thm}
$h(\cdot)$ is a rank function. 
Hence, $(E,h)$ is a polymatroid.
\end{thm}
A (poly)matroid is then called entropic, if its rank function can expressed as the entropy of a certain random vector, as above.  
A proof of the previous theorem is available in \cite{fu,lov}.
The work of Han, Fujishige, Zhang and Yeung, \cite{han,fu,zh} has resulted in the complete characterization of entropic matroids for $|E| = 2 ,3$.
However, the problem is open when $|E| \geq 4$. Note that in our case, where we have been interested in characterizing BUMAC matroids instead of entropic matroids, we have also faced a different phenomenon when $|E| \geq 4$.
Other similar problems have been studied in \cite{matus}.



\section{Main result: polar codes for MACs}\label{msol}
In this section, we describe our polar code construction for the MAC and prove the main theorem of the paper. 


Let $n=2^l$ for some $l \in \mZ_+$ and  
let $G_n=    \bigl[\begin{smallmatrix} 
      1 & 0 \\
      1 & 1 \\
   \end{smallmatrix}\bigr]^{\otimes l}$ denote the $l$-th Kronecker power of the given matrix. Let $U[k]^n:=(U_1[k],\ldots,U_n[k])$ and
\begin{align*}
& X[k]^n = U[k]^n G_n, \quad k \in \M.
\end{align*}

When $X[\M]^n$ is transmitted over $n$ independent uses of $P$ to receive $Y^n$, define for any $i \in \{1,\ldots,n\}$ the channel 
\begin{align}
P_{(i)}: \F_2^m \rightarrow \Y^n \times \F_2^{m(i-1)} \label{pi}
\end{align}
to be the channel whose inputs and outputs are $U_i[\M] \rightarrow Y^n U^{i-1}[\M]$.

\def\mB{\mathbb{B}}

Let $n \geq 1$ and $\e_n>0$, classify each $P_{(i)}$ as either `polarized' or `not polarized' according to the function $I(P_{(i)})$ being valued within $\e_n$ of $\mZ$ or not.
(We will choose an appropriate sequence $\{\e_n\}$ below.  For the moment note only that by Theorem \ref{mat}, if $\e_n$ were any fixed constant, the channels $P_{(i)}$ are in the `polarized' category except for a vanishing fraction of indices $i$.)  For $i$ for which $P_{(i)}$ is in the polarized category, set $r_i$ to be the integer within $\e_n$ of $I(P_{(i)})[\M]$.  Theorem \ref{with2} let us conclude the existence of a $r_i\times m$ matrix $A_i$ for which $H\bigl(A_iU_i[\M] \bigm| Y^n U^{i-1}[\M]\bigr)<\gamma(\e_n)$, that is to say the output of channel $P_{(i)}$ determines $A_i U_i[\M]$ with high probability\footnote{Indeed, by Problem 4.7 in \cite{gallager}, with probability at least $1-\gamma(\e_n)$.}.

We now describe what we refer to as the polar encoder and decoder for the MAC.  The encoder will be specified via the sets $B_i\subset \M$, the set of users sending data on $P_{(i)}$.  These will be chosen as follows: If $P_{(i)}$ is not polarized $B_i$ is empty.  Otherwise, select $r_i$ linearly independent columns of the matrix $A_i$, and put $k$ in $B_i$ if and only if the $k$'th column is selected.  For a user $k$ let $\mathcal G[k]$ be the set of $i$ for which $k\in B_i$.  For each user $k$ and $i\not\in \mathcal G[k]$ choose $U_i[k]$ independently and uniformly at random, reveal all these `frozen' choices to user $k$ and also to the decoder.  The encoder for user $k$ will transmit uncoded bits on channels included in $\mathcal G[k]$, on the other channels it will transmit the frozen values.

The decoder operates by successively decoding $U_1[\M]$, $U_2[\M]$, \dots, $U_n[\M]$.  At stage $i$, having
already decoded $U^{i-1}[\M]$ (assume correctly, for the moment), it is in possession of $(Y^n,U^{i-1}[\M])$, the output of $P_{(i)}$.  It can thus determine $A_i U_i[\M]$ with high probability.  Since it knows $U_i[B_i^c]$, it can
determine $\sum_{k\in B_i} A_i[k]U_i[k]$, and as $\{A_i[k]:k\in B_i\}$ are linearly independent, it can determine $U_i[B_i]$.

Observe that for this decoder to operate as described above, it needs the aid of a genie which
provides it with $U^{i-1}[\M]$ at stage $i$ of the decoding.  Let $\hat U_i[\M]=\phi_i(Y^n,U^{i-1}[\M])$ denote the decoding function of such a decoder.  Observe now, that if we construct an unaided decoder via $\tilde U_i[\M]=\phi_i(Y^n,\tilde U^{i-1}[\M])$ using the same decoding function of the genie-aided decoder, the block error event for this unaided decoder
$\tilde U^n[\M]\neq U^n[\M]$ is the same as the block error event $\hat U^n[\M]\neq U^n[\M]$ of the genie aided decoder.  Thus, the block error probability of the unaided decoder $P_e(n)$ is equal to the block error probability of the genie aided decoder and so can be upper bounded as
$$
P_e(n) \leq \sum_i P_e(P_{(i)}, A_i U_i)
$$
where $P_e(P_{(i)}, A_i U_i)$ is the probability of error in determining $A_iU_i$ from the output of the channel $P_{(i)}$.  Note now, that we have to be careful in our choice of $\e_n$: we need to take $\e_n$ small enough to ensure that $nP_e(P_{(i)}, A_i U_i)$ is small.  We will see in the proof of Theorem 9 that channel polarization happens so rapidly that even with such a more stringent choice of $\e_n$ the fraction of non polarized channels vanishes with increasing $n$.   (Indeed, for any $\beta<1/2$, one can choose $\e_n=2^{-n^\beta}$ and still ensure polarization.)

Since $I[\M] (P)$ is preserved through the polarization process (cf. the equality in \eqref{sup}), we guarantee that with $ \delta_n$ denoting the fraction of unpolarized channels,
$$ \text{Sum-Rate} (n):= \frac{1}{n} \sum_{k=1}^m |\mathcal{G}[k]| > I[\M] (P) - \delta_n-\e_n,$$
Thus if $\e_n\to0$ is chosen so that $\delta_n\to0$, the communication system described above achieves the uniform sum rate of the underlying channel.  The question as to if $\e_n$ can be chosen so that both $\delta_n\to0$ and the block error probability decays to zero is answered in the affirmative by the theorem below.


\begin{thm}\label{main}
For any $m\geq 1$, any binary input MAC $P$ with $m$ users, and any $\beta < 1/2$, there exists an integer $n_0$ and a sequence of codes with polar encoders and decoders described above such that the probability of error for a block length $n$ satisfies
$$P_e (n) \leq 2^{-n^\beta}, \quad \forall n \geq n_0$$
and $$\lim \inf _{n \rightarrow \infty} \text{Sum-Rate} (n) \geq I[\M] (P).$$ 
\end{thm}

As for the polar code in the single-user setting \cite{ari}, the encoding and decoding complexity of this code is $O(n \log n)$.

\begin{proof}[Proof of Theorem \ref{main}]
Fix $\alpha\in(\beta,1/2)$, $\e\in(0,1/2)$ and $\e_n=2^{-n^\alpha}$.  Let $\textrm{int}(x)$ denote the closest integer to $x$ and define
\begin{align*}
\mathcal{D}_n := \{ i \in &\{1,\dots,n\} :  \text{$I(P_{(i)})[J] \in \mZ \pm \e$ for any $J$}, \\
 & \exists \text{$A_i\in \F_2^{r_i\times m}$ with $r_i=\textrm{int}(I(P_{(i)}[\M]))$ and}\\
 & I(A_iU_i[\M];Y^n U[\M]^{i-1})> r_i-\e_n\}.
\end{align*}
For $i \in \mathcal{D}_n$, we have $H\bigl(A_iU_i[\M] \bigm| Y^n U^{i-1}[\M]\bigr)<\e_n$, and the output of channel $P_{(i)}$ determines $A_i U_i[\M]$ with high probability, namely
\begin{align}
P_e(P_{(i)}, A_i U_i) \leq\e_n.
\end{align}
Therefore,
\begin{align}
P_e(n) &\leq \sum_{i \in \mathcal{D}_n} P_e(P_{(i)}, A_i U_i) \\
& \leq n \e_n=o(2^{-n^\beta}).
\end{align}
Hence, such a choice of $\e_n$ guarantees the first claim of the Theorem. We now show that such an $\e_n$ is still large enough to maintain most of the polarized MACs active, causing no loss in the sum-rate as stated in the second claim of the Theorem. To this end, we need the following definition and result.
\begin{definition}\label{[s]}
For a $m$-user BMAC $P$ with output alphabet $\Y$ and for $S \subseteq \M$, we define $P^{[S]}$ to be the single-user binary input channel with output alphabet $\Y$, obtained from $P$ by
\begin{align*}
P^{[S]} (y| s) = \frac{1}{2^{m-1}} \sum_{x[\M] \in \F_2^m: \sum_{i \in S} x_i = s  } P (y| x[\M])
\end{align*}
for all $y \in \Y$, $s \in \F_2$.
Schematically, if $P: X[\M] \rightarrow Y$, we have $P^{[S]}: \sum_{i \in S} X_i \rightarrow Y$.
\end{definition}
\begin{lemma}\label{speeds}
Let $P_{(i)}$ be the channel defined in \eqref{pi} and let $(P_{(i)})^{[S]}$ be the corresponding single-user channel (cf. Definition \ref{[s]}). We have for any $\e>0$, $\alpha<1/2$ and $S \subseteq E_m$
\begin{align*}
&\lim_{l \rightarrow \infty} \frac{1}{n} |\{ i \in \{1,\ldots, n\} : I((P_{(i)})^{[S]}) > 1-\e ,\notag \\ 
&\qquad I((P_{(i)})^{[S]}) < 1-2^{-n^{ \alpha}} \}| = 0. 
\end{align*}
\end{lemma}
The proof of this lemma is given below.  Let
\begin{align}
& D_n[S] : = \{ i \in \{1,\dots,n\}: I((P_{(i)})^{[S]}) > 1-\e_n\},\\
& \widetilde{D}_n[S] : = \{ i \in \{1,\dots,n\}: I((P_{(i)})^{[S]}) > 1-\e \}.
\end{align}
From Lemma \ref{speeds}, we have that
\begin{align}
&\max_{S \in 2^{\M}} \frac{1}{n} | \widetilde{D}_n[S] \setminus D_n[S] | \rightarrow 0.
\end{align}
This implies that
\begin{align}
&\frac{1}{n} | \widetilde{\mathcal{D}}_n \setminus \mathcal{D}_n | \rightarrow 0 \label{gap}
\end{align}
where
\begin{align*}
\widetilde{\mathcal{D}}_n := \{ i \in &\{1,\dots,n\} :  \text{$I(P_{(i)})[J] \in \mZ \pm \e$ for any $J$},\\
 & \exists \text{$A_i\in \F_2^{r_i\times m}$ with $r_i=\textrm{int}(I(P_{(i)}[\M]))$ and}\\
 & I(A_iU_i[\M];Y^n U[\M]^{i-1})> r_i-\gamma(\e)\}
\end{align*}
where $\gamma(\e)$ is as in Theorem \ref{with2}.  (The only difference between $\mathcal D$ and $\widetilde{\mathcal D}$ is in the $\gamma(\e)$ and $\e_n$ in the last line.)

Since from Theorem \ref{with2} 
\begin{align*}
\lim_{l \rightarrow \infty} \frac{1}{n} |\widetilde{\mathcal{D}}_n| = 1,
\end{align*}
we also have from \eqref{gap}
\begin{align*}
\lim_{l \rightarrow \infty} \frac{1}{n} |\mathcal{D}_n| = 1.
\end{align*}
Finally, since the polarization process preserves the sum-rate, we conclude the proof of the Theorem.
\end{proof}

\begin{proof}[Proof of Lemma \ref{speeds}]
Note that 
\begin{align*}
&(P^{[S]})^- \equiv  (P^{-})^{[S]}  \\
& (P^{[S]})^+ \preceq  (P^{+})^{[S]} 
\end{align*}
where $\equiv$ means that the two transition probability distributions are the same and where $\preceq$ means that they are degraded in the sense 
$$P_1(y | x) \preceq P_2(y|x) \text{ if } P_1(y|x)=P_2(\phi (y)|x) $$
for some function $\phi$. 
Hence, 
defining the {\it Bhattacharyya parameter} of a single-user channel $Q$ with binary input and output alphabet $\Y$ by 
$$Z(Q) = \sum_{y\in \Y} \sqrt{Q(y|0) Q(y|1)},$$ 
we have
\begin{align*}
&  Z[ (P^{-})^{[S]} ]= Z [ (P^{[S]})^- ]  \leq 2 Z[ P^{[S]} ]\\
&   Z[  (P^{+})^{[S]} ] \leq Z[ (P^{[S]})^+ ]=Z[  P^{[S]} ]^2   
\end{align*}
and the random process $Z_\ell=Z [ (P_\ell)^{[S]} ]$ satisfies 
\begin{align}
&Z_{\ell+1} \leq Z_{\ell}^2 \text{ if } B_{\ell+1}=1, \label{z1} \\
&Z_{\ell+1} \leq 2 Z_\ell  \text{ if } B_{\ell+1}=0. \label{z2}
\end{align}
We then conclude by using Theorem 3 of \cite{ari2}, which shows that a random process which satisfies\footnote{the conditions required in Theorem 3 of \cite{ari2} are indeed weaker than what we have here} \eqref{z1} and \eqref{z2} satisfies for any $\beta < 1/2$
$$\lim \inf_{\ell \rightarrow \infty} \pp(Z_\ell \leq 2^{-2^{\beta \ell}}) \geq \pp(Z_\infty=0).$$
Hence, we have proved that
\begin{align*}
&\lim_{l \rightarrow \infty} \frac{1}{2^l} |\{ i \in \{1,\ldots, 2^l\} : I((P_{(i)})^{[S]}) > 1-\e ,\notag \\ 
&\qquad Z((P_{(i)})^{[S]}) \geq 2^{-2^{l \beta}} \}| = 0. 
\end{align*}
To conclude the proof of the lemma, we use the fact that for any binary input discrete memoryless channel $Q$, 
we have $I(Q) +Z(Q)>1$, hence $I(Q)< 1-\delta$ implies $Z(Q)> \delta$.
\end{proof}

\section{Coding for the AWGN channel}
We can use the results of Section \ref{msol} to construct capacity-achieving codes for the AWGN channel in the following way. Over an AWGN channel, by transmitting the standardized average of i.i.d. binary random variables, scaled to satisfy the power constraint, the receiver observes
$$Y = \frac{2\sqrt{p}}{\sqrt{m}} \sum_{i=1}^m (X_i-1/2) + Z,$$
where $Z$ is Gaussian distributed. We can view this channel as being a $m$-user BMAC, $(X_1,\ldots,X_m) \rightarrow Y$, and the polar code constructed in this paper can be used to communicate over this channel. From the central limit theorem, by taking $m$ arbitrarily large, the input distribution of previous scheme is arbitrarily close to a Gaussian distribution, and hence, this coding scheme can achieve rates arbitrarily close to the AWGN capacity. To ensure that this scheme provides a `low encoding and decoding complexity code' for the AWGN channel, one has to make further complexity considerations when assuming $m$ arbitrarily large. First, the decoder must recover a $m$-dimensional binary vector over each extremal MAC and the total (maximal) number of hypothesis is $2^m$. For this, the decoder can proceed with each of the $m$ users individually (reducing the problem to $m$ successive hypothesis tests), by using the marginalized single-user channel between one user and the output, which is an extremal channel in the single-user sense. Also, one maximal independent set of users needs to be identified for each extremal MAC, to know where the information bits should be sent. There is no need to check exponentially many sets for this purpose, since this is achieved in at most $m$ steps, by using a greedy algorithm that checks the independence of a given set and increases the set by one element at each step (starting with the empty set).

\section{Discussion}
We have constructed a polar code for the MAC with arbitrarily many users, 
which preserves the properties (complexity, error probability decay) of the polar code constructions in \cite{ari, 2mac}. 
The polarization technique brings an interesting perspective on the MAC problem: by polarizing the MACs for each user separately, we create a collection of extremal MACs which are ``trivial'' to communicate over, both regarding how to handle noise (noiseless or pure noise) but also regarding how to handle interference (which is, modulo synchronization in the code, removed). We have also shown that the extremal MACs are in a one-to-one correspondence with the linear deterministic MACs, i.e., MACs whose outputs are linear forms of the inputs.  
The polar code constructed in this paper is shown to achieve only a portion of the dominant face of the MAC region, which is however sufficient to achieve the uniform sum rate on any binary input MAC. 
There are examples of non-extremal MACs where the polar code described in this paper can achieve rates in the entire uniform rate region, for example, this is the case for a 2-user MAC whose output is $X_1+X_2$ with probability half and $(X_1,X_2)$ with probability half.   
In general, this may not be the case. 
Finally, we have considered in this paper MACs with arbitrary many users but binary input alphabets for each user.
However, for a MAC with $m$ users and $q$-ary input alphabets, where $q=2^k$, we can split each user into $k$ virtual users with binary inputs and use the polar code construction of this paper to achieve the uniform sum rate.   Furthermore, if an $m$-user $q$-ary input MAC requires a certain distribution to achieve the (true) sum rate, then, we can split each user into multiple virtual users with binary inputs, map the input vector of these to the channel inputs so that the uniform binary distribution on the virtual users induces an approximation of the required distribution (which will get better with increasing number of virtual users), and thus achieve the sum capacity of an arbitrary MAC.

 


\appendix\label{app}
In this section, we prove Theorem \ref{with} and Theorem \ref{with2}. We first need an auxiliary lemma.
\begin{lemma}\label{m2}
Let $W$ be a binary MAC with 2 users. Let $X[E_2]$ with i.i.d. uniform binary components and let $Y$ be the output of $W$ when $X[E]$ is sent.
If $I(X[1];YX[2])$, $I(X[2];YX[1])$ and $I(X[1]X[2];Y)$ have specified integer values, then $I(X[1];Y), I(X[2];Y)$ and $I(X[1]+X[2];Y)$ have specified values in $\{0,1\}$.
\end{lemma}
\begin{proof}
Let 
\begin{align*}
&I:=[I(X[1];YX[2]), I(X[2];YX[1]), I(X[1]X[2];Y)]\\
&J:=[I(X[1];Y), I(X[2];Y), I(X[1]+X[2];Y)].
\end{align*}
Note that by the polymatroid property of the mutual information, we have
\begin{align}
I \in \{[0,0,0],  [0,1,1], [1,0,1],  [1,1,1],  [1,1,2]\}. \label{poss}
\end{align}
Let $y \in \supp (Y)$ and for any $x\in \F_2^2$ define $\pp(x|y) = W(y|x)  / \sum_{z \in \F_2^2}W(y|z)$ (recall that $W$ is the MAC with inputs $X[1],X[2]$ and output $Y$). 
Assume w.l.o.g. that $p_0:=\pp(0,0|y)>0$. 

\begin{itemize}
\item If $I=[0,0,0]$ we clearly must have $J =[0,0,0]$.

\item If $I=[\star,1,1]$, we have $I(X[2];YX[1])=1$ and we can determine $X[2]$ by observing $X[1]$ and $Y$, which implies 
$$\pp (01|y)=0.$$
Moreover, since $I(X[1];Y)= I(X[1]X[2];Y) - I(X[2];Y X[1])=0$, i.e., $X[1]$ is independent of $Y$, we must have that $\sum_{x[2]}\pp(x[1]x[2]|y)$ is uniform, and hence, 
\begin{align*}
& \pp (00| y ) = 1/2,
& \pp (10| y ) + \pp (11| y )=1/2. 
\end{align*}
Now, if $\star = 1$, by a symmetric argument as before, we must have $\pp (11| y )=1/2$ and hence the input pairs $00$ and $11$ have each probability half (a similar situation occurs when assuming that $\pp(x|y)>0$ for $x\neq (0,0)$), and we can only recover $X[1]+X[2]$ from $Y$, i.e., $J=[0,0,1]$.
If instead $\star=0$, we then have $I(X[2];Y)= I(X[1]X[2];Y) - I(X[1];Y X[2])=1$ and from a realization of $Y$ we can determine $X[2]$, i.e., $\pp(10)=1/2$ and $J=[0,1,0]$.

\item If $I=[1,0,1]$, by symmetry with the previous case, we have $J=[1,0,0]$.

\item If $I=[1,1,2]$, we can recover all inputs from $Y$, hence $J =[1,1,1]$.
\end{itemize}
\end{proof}

\begin{proof}[Proof of Theorem \ref{with}] 
Let $I[S](W)$ be assigned an integer for any $S \subseteq \M$.
By the chain rule of the mutual information
$$I(X[\M]; Y) = I(X[S]; Y) + I(X[S^c] ; Y X[S]),$$
and we can determine $I(X[S]; Y)$ for any $S$. Since for any $T \subseteq S$
$$I(X[S]; Y) = I(X[T]; Y) + I(X[S-T] ; Y X[T]),$$ 
we can also determine $I(X[S]; Y X[T])$ for any $S,T \subseteq \M$ with $S \cap T = \emptyset$.
Hence we can determine 
\begin{align*}
&I(X[1] ,X[2]; Y X[S]) \\
& I(X[1]; Y X[S] X[2]) \\
& I(X[2] ; Y X[S] X[1] ) 
\end{align*}
and using Lemma \ref{m2}, we can determine
\begin{align*}
&I(X[1] + X[2] ; Y X[S]) 
\end{align*}
for any $S \subseteq \M$ with $\{1,2\} \notin S$, hence 
\begin{align*}
&I(X[i] + X[j] ; Y) 
\end{align*}
for any $i,j \in \M$.

Assume now that we have determined $I(\sum_T X[i] ; Y X[S])$ for any $T$ with $|T| \leq k$ and $S \subseteq \M -T$. 
Let $T=\{1,\dots, k\}$ and let $S \subseteq \{k+2,\dots,m\}$.
\begin{align*}
&I(\sum_T X[i], X[k+1]; Y X[S] ) \\&= I( X[k+1]; Y X[S]) + I(\sum_T X[i] ; YX[S] X[k+1]),
\end{align*}
in particular, we can determine
\begin{align*}
& I(X[k+1] ; Y \sum_T X[i], X[S] ) \\
&= I(\sum_T X[i] , X[k+1]; Y X[S] )\\& - I(\sum_T X[i]; YX[S] )
\end{align*}
and
\begin{align*}
&I(\sum_T X[i] ,X[k+1]; Y X[S]) \\
& I(\sum_T X[i] ; Y X[S] X[k+1]) \\
& I(X[k+1] ; Y \sum_T X[i] ,X[S] ) 
\end{align*}
and using Lemma \ref{m2}, we can determine
\begin{align*}
&I(\sum_T X[i] + X[k+1]; Y X[S])
\end{align*}
hence
\begin{align*}
&I(\sum_{T}X[i]; Y)
\end{align*}
for any $T \subseteq \M$ with $|T|= k+1$.
Hence, inducting this argument, we can determine $I(\sum_{T}X[i]; Y)$ for any $T \subseteq \M$.

Note that the values of these mutual informations must be consistent, for example, if $I(X[1]+X[2];Y)=1$ and $I(X[1]+X[3];Y)=1$, we must have $I(X[2]+X[3];Y)=1$. Hence, the $T$'s for which $I(\sum_{T}X[i]; Y)$ is assigned 1 must be in agreement with these linear relationships, which can be compactly expressed as $I(A X[\M]; Y)= \rank (A)$ for some binary matrix $A$. Finally, one can check directly (or by using Theorem \ref{mat}) that $\rank(A) = I(X[\M];Y)$.
\end{proof}

In order to prove the ``approximative'' version of Theorem \ref{with}, i.e., Theorem \ref{with2}, we need the following lemma which is a corollary of Lemma 33 in \cite{2mac}. The proof of Theorem \ref{with2} follows then from Lemma \ref{m2e} and the proof of Theorem \ref{with}.
\begin{lemma}\label{m2e}
Let $W$ be a binary MAC with 2 users. Let $X[E_2]$ with i.i.d. uniform binary components and let $Y$ be the output of $W$ when $X[E]$ is sent.
If $I(X[1];YX[2])$, $I(X[2];YX[1])$ and $I(X[1]X[2];Y)$ have specified integer values within $\e$, then $I(X[1];Y), I(X[2];Y)$ and $I(X[1]+X[2];Y)$ have specified values outside $(\gamma(\e),1-\gamma(\e))$ with $\gamma(\e) \stackrel{\e \to 0}\to 0$.
\end{lemma}





\end{document}